\documentclass{amsart}

\usepackage{amssymb,latexsym,amsfonts,amsmath}
\usepackage{graphicx,psfrag,epsfig,epsf}
\include{diagrams}
\usepackage{eso-pic}
\usepackage{graphicx}
\usepackage{color}
\usepackage{diagrams}
\usepackage{algorithm2e}
\usepackage{tikz}
\usetikzlibrary{automata}
\usepackage{subfigure}
\usepackage{amsfonts}
\usepackage{amssymb,latexsym,amsfonts,amsmath}
\usepackage{graphicx}

\newtheorem{theorem}{Theorem}[section]
\newtheorem{lemma}[theorem]{Lemma}

\newtheorem{proposition}[theorem]{Proposition}
\newtheorem{corollary}[theorem]{Corollary}

\theoremstyle{definition}
\newtheorem{definition}[theorem]{Definition}
\newtheorem{example}[theorem]{Example}

\theoremstyle{remark}

\numberwithin{equation}{section}

\newcommand{\Pre}{\mathrm{Pre}}
\newcommand{\iso}{\mathrm{iso}}

\begin{document}

\begin{abstract}                          % Abstract of not more than 200 words.
Finite state machines are widely used as a sound mathematical formalism that appropriately describes large scale, distributed and complex systems. 
Multiple interactions of finite state machines in complex systems are well captured by the notion of non--flat systems. Non--flat systems are "finite state machines" where each "state" can be either a basic state or an aggregate of finite state machines. 
By expanding a non--flat system, a flat system is obtained which is an ordinary finite state machine. 
In this paper we introduce a novel class of non--flat systems called Arena of Finite State Machines (AFSM). AFSMs are collections of finite state machines that interact concurrently through a communication network. We propose a notion of equivalence, termed compositional bisimulation, that allows the complexity reduction of AFSMs without the need of expanding them to the corresponding FSMs. The computational complexity gain obtained from this approach is formally quantified in the paper. An application of the proposed framework to the regulation of gene expression in the bacterium \textit{Escherichia coli} is also presented.
\end{abstract}

\title[Arenas of Finite State Machines]{Arenas of Finite State Machines}
\thanks{The research leading to these results has been partially supported by the Center of Excellence DEWS and received funding from the European Union Seventh Framework Programme [FP7/2007-2013] under grant agreement n.257462 HYCON2 Network of excellence.}

\author[Giordano Pola, Maria D. Di Benedetto and Elena De Santis]{
Giordano Pola$^{1}$, Maria D. Di Benedetto$^{1}$ and Elena De Santis$^{1}$}
\address{$^{1}$
Department of Electrical and Information Engineering, Center of Excellence DEWS,
University of L{'}Aquila, Poggio di Roio, 67040 L{'}Aquila, Italy}
\email{ \{giordano.pola,mariadomenica.dibenedetto,elena.desantis\}@univaq.it}

\maketitle

\section{Introduction}
Finite state machines (FSMs) are widely used in modeling complex systems ranging from computer and communication networks, automated manufacturing systems, air traffic management systems, distributed software systems, among many others, see e.g. \cite{cassandras,ModelChecking}. 
%Finite state machines have been also employed as a sound mathematical paradigm to describe purely continuous and hybrid systems in the context of the so--called correct--by--design embedded control software, see e.g. \cite{DiscAbs,paulo,GirardTAC2010,Belta:06} and the references therein. 
The increasing complexity of large scale systems demanded during the years for formal methods that can render their analysis tractable from a computational complexity point of view. Several approaches have been proposed in the literature, which include abstraction, modular verification methods, symmetry and partial order reduction, see e.g. \cite{ModelChecking}. The common goal of these approaches is to find an FSM that is equivalent to the original one, but with a set of states of smaller size. 
In this paper we follow the approach by Alur and co--workers (see e.g. \cite{HFSM1,CHSM}), where a complex system is viewed as a "non--flat" system. 
A non--flat system is a "finite state machine" where each "state" can be either a basic state or a superstate \cite{statecharts} that hides inside an FSM or even a composition of FSMs. By expanding the superstates of a non--flat system to their corresponding FSMs an ordinary FSM is obtained. One of the early non--flat systems that appeared in the literature are \textit{hierarchical state machines} (HSMs) \cite{HFSM1}. 
%Reachability problems for HSMs can be studied with time complexity that scales linearly with the size of HSMs. 
While HSMs well capture modeling features of many design languages as for example Statecharts \cite{statecharts}, they only consider sequential interaction among the FSMs involved. \textit{Recursive state machines} (RSMs) \cite{RFSM} extend HSMs by allowing recursion in the sequential interaction of FSMs. As such, they well model sequential programming languages with recursive procedure calls. 
%Reachability and cycle detection problems for RSMs can be studied with space and time complexity that scales linearly with the size of RSMs. 
%Due to recursion, the underlying global state-space is infinite and behaves like a pushdown system. 
\textit{Recursive Game Graphs}, a natural adaptation of RSMs to a game theoretic setting, have been studied in \cite{Etessami2004}; \textit{Pushdown Graphs} have been studied in \cite{Cachat2002a}. 
Both HSMs and RSMs do not exhibit concurrent compositional features. \textit{Communicating hierarchical state machines} (CHSMs) \cite{CHSM} generalize HSMs, by allowing FSMs to interact not only sequentially but also concurrently, through the notion of parallel composition. %A naive approach to analyze and control such systems is to expand them to ordinary FSMs, thus incurring in an exponential grow of the set of states. This method is in general demanding from the computational complexity point of view. 
%A challenge in this research direction is to derive formal methods for the analysis of non--flat systems, by directly exploiting their inherent hierarchical compositional structure. 
Reachability problems and checking language and bisimulation equivalences for CHSMs are proven in \cite{CHSM} to fall in the class of exponential time and space complexity problems. This complexity result is in line with the ones further established in \cite{StateExplosion2,StateExplosion1} on  complexity arising in checking a range of equivalence notions in the linear time--branching time spectrum \cite{Spectrum} for networks of FSMs, modeled by parallel composition of FSMs. 
%In particular, these papers showed that checking any equivalence relation lying between bisimilarity and trace preorder is an exponential time hard problem, as conjectured in \cite{StateExplosion0}. 
By following the conjecture of Rabinovich in \cite{StateExplosion0}, the work in \cite{StateExplosion2,StateExplosion1} strongly suggest that there is no way to escape the so--called state explosion problem, when checking behavioral relations and in particular bisimulation equivalence, for non--flat systems exhibiting concurrent--types interaction. \\
In this paper we identify a novel class of concurrent non--flat systems, termed \textit{Arenas of Finite State Machines} (AFSMs) for which complexity reduction via bisimulation can be performed without incurring in the state explosion problem. AFSMs are collections of FSMs that interact concurrently, through a communication network. For AFSMs we propose a notion of equivalence, termed \textit{compositional bisimulation}, that is based on the communication network governing the interaction mechanism among the FSMs. The main contribution of the paper resides in showing that compositional bisimulation equivalence between AFSMs implies bisimulation equivalence between the corresponding expanded FSMs. 
This result is important because it implies that all properties preserved by bisimulation equivalence, e.g. linear temporal logic properties \cite{ModelChecking}, are also preserved by compositional bisimulation. Therefore, it can be of help in the formal verification and control design of complex systems modeled by AFSMs that admit compositional bisimulation. 
%Since existence of compositional bisimulation between a pair of AFSMs $\mathbb{A}^{1}$ and $\mathbb{A}^{2}$ implies bisimulation equivalence between the corresponding expanded FSMs $\mathbb{M}(\mathbb{A}^{1})$ and $\mathbb{M}(\mathbb{A}^{2})$ 
%all properties preserved by bisimulation equivalence between the FSMs $\mathbb{M}(\mathbb{A}^{1})$ and $\mathbb{M}(\mathbb{A}^{2})$ obtained by expanding the AFSMs $\mathbb{A}^{1}$ and $\mathbb{A}^{2}$, e.g. linear temporal logic properties \cite{ModelChecking}, are also preserved by compositional bisimulation between AFSMs $\mathbb{A}^{1}$ and $\mathbb{A}^{2}$. Therefore, the proposed results can be of help in the formal verification and control design of complex systems modeled by AFSMs that admit compositional bisimulation. 
%The notion of compositional bisimulation provides a method to escape the state explosion problem. 
A computational complexity analysis reported in the paper reveals that 
%Existence of compositional bisimulations allows checking bisimulation equivalence of AFSMs, without the need of expanding them to ordinary FSMs and hence, without incurring in the aforementioned state explosion problem. 
checking compositional bisimulation between a pair of AFSMs scales as $O(N_{1}^{2}+N_{2}^{2})$ in space complexity and as $O((N_{1}^{2}+N_{2}^{2})\ln(N_{1}+N_{2}))$ in time complexity, with the numbers $N_{1}$ and $N_{2}$ of FSMs composing the AFSMs.  
A standard approach, based on expanding the AFSMs to the corresponding FSMs exhibits an exponential space and time complexity. 
%, as discussed in the paper. 
%An explicit complexity analysis in checking bisimulation between FSMs, obtained by expanding AFSMs, is also given which shows exponential space and time complexity, in accordance to the literature on non--flat systems. 
%Checking bisimulation equivalence of the AFSMs by expanding them to FSMs is shown to be exponential hard problem. 
%A computational complexity analysis is discussed, which quantifies the gain of the proposed approach. 
%A preliminary investigation on AFSMs appeared in \cite{PolaIFAC2011}. In this paper we present a detailed and mature description of the results announced in \cite{PolaIFAC2011}, which includes proofs. 
% and an example in the context of the regulation of gene expression in the bacterium \textit{Escherichia Coli}  \cite{Alon,BiocBook,BiocBook1}. 
An application of the proposed results to the modeling and complexity reduction of the regulation of gene expression in the single--celled bacterium \textit{E. coli} is included.

\section{Preliminary definitions}
\label{Sec:FSM}
\subsection{Notation}
Given a set $A$, the symbol $2^{A}$ denotes the set of subsets of $A$ and the symbol $|A|$ denotes the cardinality of $A$. 
If $|A|=1$ then $A$ is said a singleton. 
A relation $R\subseteq A\times B$ is said to be total if for any $a\in A$ there exists $b\in B$ such that $(a,b)\in R$ and conversely, for any $b\in B$ there exists $a\in A$ such that $(a,b)\in R$. Given a relation $R\subseteq A\times B$, the inverse of $R$, denoted $R^{-1}$, is defined as $\{(b,a)\in B\times A : (a,b)\in R\}$. A relation $R\subseteq A\times B$ is the identity relation if $A=B$ and $a=b$ for all $(a,b)\in R$. %We denote by $R(a)$ and $R^{-1}(b)$ the sets $\{a\in A | (a,b)\in R\}$ and $\{b\in B | (a,b)\in R\}$. 
A directed graph is a tuple $G=(V,E)$ where $V$ is the set of vertices and $E$ is the set of edges. %Any set involved in the sequel is assumed to be embedded in the universal set $\mathcal{U}$.
We denote by $\mathbb{N}$ the set of positive integers. 

\subsection{Finite State Machines}
In this paper we consider finite state machines in the formulation of Moore \cite{Moore} where states are labeled with outputs and transitions are labeled with inputs.
\begin{definition} 
\cite{FSMmin} 
\label{DefFSM}
A Finite State Machine (FSM) is a tuple 
\begin{equation}
\label{FSMtuple}
M=(X,x^{0},U,Y,H,\Delta),
\end{equation}
where $X$ is a finite set of states, $x^{0}\in X$ is the initial state, $U$ is a finite set of input symbols, $Y$ is a finite set of output symbols,  $H:X \rightarrow 2^{Y}$ is an output map, and $\Delta \subseteq X\times 2^{U} \times X$ is a transition relation. 
%where:
%\begin{itemize}
%\item $X$ is a finite set of states;
%\item $x^{0}\in X$ is the initial state;
%\item $U$ is a finite set of input symbols;
%\item $Y$ is a finite set of output symbols;
%\item $H:X \rightarrow 2^{Y}$ is an output map;
%\item $\Delta \subseteq X\times 2^{U} \times X$ is a transition relation. 
%\end{itemize}
\end{definition}
When $x^{0}$ is skipped from the tuple in (\ref{FSMtuple}) any state in $X$ is assumed to be an initial state. 
We denote a transition $(x,u,x^{\prime})\in\Delta$ of FSM $M$ by $x \rTo^{u}_{\Delta} x^{\prime}$.  
By definition of $\Delta$, a transition of the form $x \rTo^{\varnothing}_{\Delta} x^{\prime}$ is allowed. Such a transition is viewed as private or internal to $M$. Throughout the paper we refer to an input $u=\varnothing$ as internal, and an input $u\neq\varnothing$ as external to $M$. 
Analogously, for a state $x\in X$, $H(x)=\varnothing$ is allowed, meaning that state $x$ is not visible from the external environment. 
Despite classical formulations of Moore machines that model the transition relation as $\Delta\subseteq X\times U \times X$ and the output function as $H:X \rightarrow Y$, we model here $\Delta$ as a subset of $X\times 2^{U} \times X$ and $H$ as a function from $X$ to $2^{Y}$. By this choice, multiple interactions of FSMs can be considered, as illustrated in Example \ref{exex} on a simple distributed system. 

\subsection{Equivalence notions}
Several notions of equivalence have been proposed for the class of finite state machines, see e.g. \cite{Spectrum}. 
In this paper we focus on the notion of bisimulation equivalence \cite{Milner,Park} that 
%Bisimulation equivalence 
is widely used as an effective tool to mitigate complexity of verification and control design of large scale complex systems, see e.g. \cite{ModelChecking}. 
%Intuitively, a bisimulation relation between a pair of FSMs $M_{1}$ and $M_{2}$ is a relation between the corresponding sets of states explaining how a state run of $M_{1}$ can be transformed into a state run of $M_{2}$ and vice versa. 
Consider a pair of FSMs $M_{i}=(X_{i},x^{0}_{i},U_{i},Y_{i},H_{i},\Delta_{i})$ ($i=1,2$). 
We start by recalling the notion of isomorphism. 
\begin{definition}
\label{DefIso}
The FSMs $M_{1}$ and $M_{2}$ are isomorphic, denoted $M_{1} \cong^{\iso} M_{2}$, if there exists a bijective function $\mathcal{T}:X_{1}\rightarrow X_{2}$ such that 
$x^{0}_{2}=\mathcal{T}(x^{0}_{1})$, $H_{1}(x_{1})=H_{2}(\mathcal{T}(x_{1}))$ for any $x_{1}\in X_{1}$, and $x_{1}\rTo^{u}_{\Delta_{1}} x^{\prime}_{1}$ if and only if $\mathcal{T}(x_{1})\rTo^{u}_{\Delta_{2}} \mathcal{T}(x^{\prime}_{1})$. 
%\begin{itemize}
%\item $x^{0}_{2}=\mathcal{T}(x^{0}_{1})$;
%\item for any $x_{1}\in X_{1}$, $H_{1}(x_{1})=H_{2}(\mathcal{T}(x_{1}))$;
%\item $x_{1}\rTo^{u}_{\Delta_{1}} x^{\prime}_{1}$ if and only if $\mathcal{T}(x_{1})\rTo^{u}_{\Delta_{2}} \mathcal{T}(x^{\prime}_{1})$. 
%\end{itemize}
\end{definition}
The notion of isomorphism is an equivalence relation on the class of FSMs. 
The notion of bisimulation equivalence is reported hereafter.
\begin{definition}
\label{Bis} 
A set $R\subseteq X_{1}\times X_{2}$ is a bisimulation relation between $M_{1}$ to $M_{2}$ if for any $(x_{1},x_{2})\in R$,
\begin{itemize}
\item[(i)] $H_{1}(x_{1})=H_{2}(x_{2})$;
\item[(ii)] existence of $x_{1} \rTo^{u_{1}}_{\Delta_{1}} x_{1}^{\prime}$ implies existence of $x_{2} \rTo^{u_{2}}_{\Delta_{2}} x_{2}^{\prime}$ such that $u_{1}=u_{2}$ and $(x_{1}^{\prime},x_{2}^{\prime})\in R$;
\item[(iii)] existence of $x_{2} \rTo^{u_{2}}_{\Delta_{2}} x_{2}^{\prime}$ implies existence of $x_{1} \rTo^{u_{1}}_{\Delta_{1}} x_{1}^{\prime}$ such that $u_{1}=u_{2}$ and $(x_{1}^{\prime},x_{2}^{\prime})\in R$.
\end{itemize}
%FSM $M_{1}$ is simulated by FSM $M_{2}$, denoted $M_{1} \preceq M_{2}$, if there exists a simulation relation from $M_{1}$ to $M_{2}$.
%\end{definition}
%\begin{definition}
%\label{Bis2} 
%Given a pair of FSMs $M_{i}=(X_{i},x^{0}_{i},U_{i},$ $Y_{i},H_{i},\Delta_{i})$ ($i=1,2$), a r
FSMs $M_{1}$ and $M_{2}$ are bisimilar, denoted $M_{1} \cong M_{2}$, if 
\begin{itemize}
\item[(iv)] \mbox{$(x^{0}_{1},x^{0}_{2})\in R$}.
\end{itemize}
\end{definition}

When the initial states $x^{0}_{1}$ and $x^{0}_{2}$ are skipped from the tuples $M_{1}$ and $M_{2}$, condition (iv) is replaced by requiring relation $R$ to be total. 
Bisimulation equivalence is an equivalence relation on the class of FSMs. 
The maximal bisimulation relation between FSMs $M_{1}$ and $M_{2}$ is a bisimulation relation $R^{\ast}(M_{1},M_{2})$ such that $R\subseteq R^{\ast}(M_{1},M_{2})$ for any bisimulation relation $R$ between $M_{1}$ and $M_{2}$. The maximal bisimulation relation exists and is unique. Given an  FSM $M$ the set $R^{\ast}(M,M)$ is an equivalence relation on the set of states of $M$. 
%We recall hereafter the notion of quotient \cite{ModelChecking} of an FSM $M$ induced by the equivalence relation $ R^{\ast}(M,M)$.
%\begin{definition}
%\label{def:quotient}
%The quotient of an FSM $M=(X,x^{0},U,$ $Y,$ $H,\Delta)$ induced by $R^{\ast}(M,M)$ is the FSM $M^{\ast}=(X^{\ast},x^{0,\ast},U^{\ast},Y^{\ast},H^{\ast},\Delta^{\ast})$, where:
%\begin{itemize}
%\item $X^{\ast}=\{C_{1},C_{2},...,C_{N}\}$, where $C_{i}$ are the equivalence classes induced by $R^{\ast}(M,M)$ on $X$;
%\item $x^{0,\ast}=\{x^{0}\}$;
%\item $U^{\ast}=U$;
%\item $Y^{\ast}=Y$;
%\item $H^{\ast}:X^{\ast}\rightarrow Y^{\ast}$ is defined by $H^{\ast}(C_{i})=y$, if $H(x)=y$ for any 
%$x\in C_{i}$;
%\item $\Delta^{\ast}\subseteq X^{\ast} \times 2^{U^{\ast}} \times X^{\ast}$ is defined by $C_{i}\rTo_{\Delta^{\ast}}^{u} C_{j}$, if 
%$x\rTo_{\Delta}^{u} x'$ for any $x\in C_{i}$ and $x'\in C_{j}$. 
%\end{itemize}
%\end{definition}
The quotient of $M$ induced by $R^{\ast}(M,M)$, denoted $\mathbf{M}_{\min}(M)$, is the FSM bisimilar to $M$ with the minimal number of states \cite{ModelChecking}. FSM $\mathbf{M}_{\min}(M)$ exists and is unique up to isomorphisms.
\begin{lemma}
\label{lemma}
%Given %a pair of FSMs 
%$M_{i}=(X_{i},x^{0}_{i},U_{i},$ $Y_{i},H_{i},\Delta_{i})$ ($i=1,2$), if 
If $\mathbf{M}_{\min}(M_{1}) \cong \mathbf{M}_{\min}(M_{2})$ then $M_{1}\cong^{\iso} M_{2}$.
\end{lemma}
\begin{proof}
Let $X_{i}$ be the set of states of $M_{i}$. 
%It is easy to see by contradiction arguments that m
Minimality of $\mathbf{M}_{\min}(M_{1})$ and $\mathbf{M}_{\min}(M_{2})$ implies that the maximal bisimulation relation $R^{\ast}$ between $\mathbf{M}_{\min}(M_{1})$ and $\mathbf{M}_{\min}(M_{2})$ is such that for any $x_{1}\in X_{1}$ and $x_{2}\in X_{2}$, sets $R^{\ast}(x_{1})=\{x_{2}\in X_{2} | (x_{1},x_{2})\in R^{\ast}\}$ and $(R^{\ast})^{-1}(x_{2})=\{x_{1}\in X_{1} | (x_{1},x_{2})\in R^{\ast}\}$ are singletons. Hence, define function $\mathcal{T}: X_{1}\rightarrow X_{2}$ by $\mathcal{T}(x_{1})=x_{2}$ when $R^{\ast}(x_{1})=\{x_{2}\}$. It is easy to see that function $\mathcal{T}$ satisfies the properties required in Definition \ref{DefIso}.
\end{proof}

We conclude this section by recalling space and time complexity in checking bisimulation equivalence between FSMs. 
%Consider a pair of FSMs $M_{i}=(X_{i},x^{0}_{i},U_{i},Y_{i},H_{i},\Delta_{i})$ ($i=1,2$). 
\begin{proposition}
\label{prop21}
\cite{BisAlg}
Space complexity in checking $M_{1}\cong M_{2}$ is $O(|X_{1}|+|\Delta_{1}|+|X_{2}|+|\Delta_{2}|)$.
\end{proposition}

\begin{proposition}
\label{prop22}
\cite{BisAlg}
Time complexity in checking $M_{1}\cong M_{2}$ is $O((|\Delta_{1}|+|\Delta_{2}|)\ln(|X_{1}|+|X_{2}|))$. 
\end{proposition}

\section{Arenas of finite state machines}
%One approach to capture multiple interaction of FSMs in a complex system is to resort to the notion of non--flat systems, see e.g. \cite{HFSM1,CHSM}. 
In this section we introduce a new class of non--flat systems \cite{HFSM1,CHSM}, called \textit{Arenas of Finite State Machines} (AFSMs). AFMSs 
are collections of FSMs that interact concurrently through a communication network. 
The syntax of an AFSM is specified by a directed graph: 
\[
\mathbb{A}=(\mathbb{V},\mathbb{E}), 
\]
where $\mathbb{V}$ is a collection of $N$ FSMs $M_{i}$ and $\mathbb{E}\subseteq \mathbb{V} \times \mathbb{V}$ describes the communication network of the FSMs $M_{i}$. 
%\begin{itemize}
%\item $\mathbb{V}$ is a collection of $N$ FSMs $M_{i}=(X_{i},x^{0}_{i},U_{i},Y_{i},H_{i},\Delta_{i})$ ($i=1,2,...,N$);
%\item $\mathbb{E}\subseteq \mathbb{V} \times \mathbb{V}$ describes the communication network of the FSMs $M_{i}$. 
%\end{itemize}
In the definition of $\mathbb{E}$ self loops $(M_{i},M_{i})\in \mathbb{E}$ would model communication of $M_{i}$ with itself, which is tautological. For this reason in the sequel we assume $(M_{i},M_{i})\notin \mathbb{E}$. 
By expanding each vertex $M_{i}\in \mathbb{V}$ of $\mathbb{A}$ an ordinary FSM is obtained, which is defined by:
\[
\mathbb{M}(\mathbb{A})=(X,x^{0},U,Y,H,\Delta),
\]
where $X=X_{1}\times X_{2} \times ... \times X_{N}$, $x^{0}=(x^{0}_{1},x^{0}_{2},...,x^{0}_{N})$, $U=\bigcup_{M_{i}\in \mathbb{V}} U_{i}$, $Y=\bigcup_{M_{i}\in \mathbb{V}}Y_{i}$, $H((x_{1},x_{2},...,x_{N}))=\bigcup_{M_{i}\in \mathbb{V}} H_{i}(x_{i})$, and $\Delta \subseteq X\times 2^{U} \times X$ is such that 
\begin{equation}
\label{ttrans}
(x_{1},x_{2},...,x_{N})\rTo_{\Delta}^{u} (x_{1}^{\prime},x_{2}^{\prime},...,x_{N}^{\prime}),
\end{equation}
%whenever the following conditions are satisfied:
%\begin{itemize}
%\item[(i)] $x_{i} \rTo_{\Delta_{i}}^{u_{i}} x^{\prime}_{i}$ is a transition of $M_{i}$, for some $u_{i}\in U_{i}$;
%\item[(ii)] $u=\bigcup_{M_{i}\in \mathbb{V}} (u_{i}\backslash (\bigcup_{M_{j}\in \Pre(\mathbb{A},M_{i})}H_{j}(x_{j})))$, where $
%\Pre(\mathbb{A},M_{i})=\{M_{j}\in \mathbb{V}\,|\,(M_{j},M_{i})\in \mathbb{E}\}$. 
%\end{itemize}
whenever $x_{i} \rTo_{\Delta_{i}}^{u_{i}} x^{\prime}_{i}$ is a transition of $M_{i}$ for some $u_{i}$ ($i=1,2,...,N$) and 
\begin{equation}
\label{u}
u=\bigcup_{M_{i}\in \mathbb{V}} (u_{i}\backslash (\bigcup_{M_{j}\in \Pre(\mathbb{A},M_{i})}H_{j}(x_{j}))),
\end{equation}
where $\Pre(\mathbb{A},M_{i})=\{M_{j}\in \mathbb{V}\,|\,(M_{j},M_{i})\in \mathbb{E}\}$. 
%where:
%\begin{itemize}
%\item $X=X_{1}\times X_{2} \times ... \times X_{N}$;
%\item $x^{0}=(x^{0}_{1},x^{0}_{2},...,x^{0}_{N})$;
%\item $U=\cup_{M_{i}\in \mathbb{V}} U_{i}$;
%\item $Y=\cup_{M_{i}\in \mathbb{V}}Y_{i}$;
%\item $H((x_{1},x_{2},...,x_{N}))=\bigcup_{M_{i}\in \mathbb{V}} H_{i}(x_{i})$;
%\item $\Delta \subseteq X\times 2^{U} \times X$ is such that 
%\begin{equation}
%\label{ttrans}
%(x_{1},x_{2},...,x_{N})\rTo_{\Delta}^{u} (x_{1}^{\prime},x_{2}^{\prime},...,x_{N}^{\prime}),
%\end{equation}
%whenever the following conditions are satisfied:
%\begin{itemize}
%\item[(i)] $x_{i} \rTo_{\Delta_{i}}^{u_{i}} x^{\prime}_{i}$ is a transition of $M_{i}$, for some $u_{i}\in U_{i}$;
%\item[(ii)] $u=\bigcup_{M_{i}\in \mathbb{V}} (u_{i}\backslash (\bigcup_{M_{j}\in \Pre(\mathbb{A},M_{i})}H_{j}(x_{j})))$, where $
%\Pre(\mathbb{A},M_{i})=\{M_{j}\in \mathbb{V}\,|\,(M_{j},M_{i})\in \mathbb{E}\}$. 
%\end{itemize}
%\end{itemize}

\begin{proposition}
Given an AFSM $\mathbb{A}$, the FSM $\mathbb{M}(\mathbb{A})$ is unique.
\end{proposition}

\begin{proof}
Entities $X$, $x^{0}$, $U$, $Y$ and $H$ in $\mathbb{M}(\mathbb{A})$ are uniquely determined from $\mathbb{A}$. For any collection of $N$ transitions $x_{i} \rTo_{\Delta_{i}}^{u_{i}} x^{\prime}_{i}$ in $M_{i}$ there exists one and only one transition in $\mathbb{M}(\mathbb{A})$ of the form (\ref{ttrans}) with $u$ uniquely specified by (\ref{u}).
\end{proof}

FSM $\mathbb{M}(\mathbb{A})$ specifies the semantics of the AFSM $\mathbb{A}$. Such a semantic is implicitly given through a composition of FSMs that can be regarded as a notion of parallel composition \cite{ModelChecking} that respects the topology of the AFMS communication network. 
%The main idea behind the above semantics is that outputs of FSMs $M_{i}$ may become inputs of FSM $M_{j}$, (only) when there exists a communication network link from $M_{i}$ to $M_{j}$, or equivalently $(M_{i},M_{j})\in\mathbb{E}$. All inputs of $M_{j}$ which are not provided by those $M_{i}$ for which $(M_{i},M_{j})\in\mathbb{E}$, become inputs that are external to $\mathbb{M}(\mathbb{A})$. 
The following simple example illustrates syntax and semantics of AFSMs.
%For the purpose of illustrating the syntax and semantics of AFSMs we present hereafter a simple example in the context of distributed computing \cite{Andrews}. 
\begin{example}
\label{exex}
Consider a distributed system composed of three computers $C_{1}$, $C_{2}$ and $C_{3}$, whose goal is to compute the Euclidean norm $\Vert z \Vert = \sqrt{z_{1}^{2}+z_{2}^{2}}$ of a vector $z=(z_{1},z_{2})\in \mathbb{R}^{2}$ in a distributed fashion. 
While $C_{1}$ and $C_{2}$ are delegated to compute respectively $z_{1}^{2}$ and $z_{2}^{2}$, $C_{3}$ takes as inputs the computations of $C_{1}$ and $C_{2}$ and outputs $\Vert z \Vert$. 
This simple distributed system can be modeled as the AFMS $\mathbb{A}=(\mathbb{V},\mathbb{E})$ where $\mathbb{V}=\{M_{1},M_{2},M_{3}\}$ and 
$\mathbb{E}=\{(M_{1},M_{3}),(M_{2},M_{3})\}$. 
FSMs $M_{i}$, each one modeling computers $C_{i}$, are illustrated in Figures\footnote{Each circle denotes a state and each edge a transition. In each circle, upper symbol denotes the state and lower symbol the output set associated with the state; symbols labeling edges denote the input sets associated with the transitions.} \ref{FSMs1}(a)(b)(c), while AFSM $\mathbb{A}$, modeling the computers' network, is depicted in Figure \ref{FSMs1}(d). 
By expanding $\mathbb{A}$, the FSM $\mathbb{M}(\mathbb{A})$ is obtained, whose accessible part%\cite{cassandras} 
\footnote{The accessible part of the FSM $M$ in (\ref{FSMtuple}) is the unique sub--finite state machine extracted from $M$, containing all and only the states of $M$ that are reachable (or equivalently, accessible) in a finite number of transitions from its initial state $x^{0}$, see e.g. \cite{cassandras}.} 
is depicted in Figure \ref{M(A)}. Starting from $(1,3,5)$, when receiving the input $\{z_{1},z_{2}\}$, FSM $\mathbb{M}(\mathbb{A})$ outputs in state $(2,4,6)$ the set $\{z_{1}^{2},z_{2}^{2}\}$ and finally in state $(1,3,7)$ the requested output $\{\Vert z \Vert\}$. For illustrating the construction of FSM $\mathbb{M}(\mathbb{A})$, we describe in detail the construction of the transition $(2,4,6) \rTo_{\Delta}^{u} (1,3,7)$. By applying the compositional rules defining the semantics of AFSMs, one gets: 
$2 \rTo^{\varnothing} 1$ is in $M_{1}$, $4 \rTo^{\varnothing} 3$ is in $M_{2}$, and $6 \rTo^{\{z_{1}^{2},z_{2}^{2}\}} 7$ is in $M_{3}$. Moreover, one first note that $\Pre(\mathbb{A},M_{1})=\Pre(\mathbb{A},M_{2})=\varnothing$ and $\Pre(\mathbb{A},M_{3})=\{M_{1},M_{2}\}$, from which $u=\varnothing$. 
%:
%\[
%\begin{array}
%{rcl}
%u 
%  & = & \varnothing \cup \varnothing \cup (\{z_{1}^{2},z_{2}^{2}\}\backslash (\{z_{1}^{2}\}\cup \{z_{2}^{2}\})) = \varnothing .
%\end{array}
%\]
%\begin{itemize}
%\item[(i)] $2 \rTo_{\Delta_{1}}^{\varnothing} 1$ is a transition of $M_{1}$, 
%$4 \rTo_{\Delta_{2}}^{\varnothing} 3$ is a transition of $M_{2}$, and 
%$6 \rTo_{\Delta_{3}}^{\{z_{1}^{2},z_{2}^{2}\}} 7$ is a transition of $M_{3}$.
%\item[(ii)] One first note that $\Pre(\mathbb{A},M_{1})=\Pre(\mathbb{A},M_{2})=\varnothing$ and $\Pre(\mathbb{A},M_{3})=\{M_{1},M_{2}\}$, from which:
%\[
%\begin{array}
%{rcl}
%u 
%  & = & \varnothing \cup \varnothing \cup (\{z_{1}^{2},z_{2}^{2}\}\backslash (\{z_{1}^{2}\}\cup \{z_{2}^{2}\})) = \varnothing .
%\end{array}
%\]
%\end{itemize}
The resulting transition $(2,4,6) \rTo^{\varnothing} (1,3,7)$ is indeed in $\mathbb{M}(\mathbb{A})$, as shown in Figure \ref{M(A)}. 
%Note that for such a transition to occur in $\mathbb{M}(\mathbb{A})$, transitions $2 \rTo_{\Delta_{1}}^{\varnothing} 1$ and $4 \rTo_{\Delta_{2}}^{\varnothing} 3$ need no external inputs, while transition $6 \rTo_{\Delta_{3}}^{\{z_{1}^{2},z_{2}^{2}\}} 7$ needs input $\{z_{1}^{2},z_{2}^{2}\}$ (external to $M_{3}$) which is provided by the outputs $\{z_{1}^{2}\}$ and $\{z_{2}^{2}\}$ of FSMs $M_{1}$ and $M_{2}$, respectively. 
\end{example}
\begin{figure}[t]
\centering
\subfigure[FSM $M_{1}$]{
\begin{tikzpicture}[->,
shorten >=0.1pt,
auto,node distance=2.0cm,thick,
inner sep=0.1 pt ,bend angle=30]
\tikzstyle{every state}=[minimum size=8mm]
\tikzstyle{every node}=[font=\footnotesize]
\tikzstyle{initial}=[pin={[pin distance=5mm,pin edge={<-,shorten <=1pt}]left:$$}]
\node[state, thin, initial] (q_1) {$\frac{\text{  }1\text{  }}{\varnothing}$ };
\node[state, thin] (q_2) [ right of=q_1] {$\frac{\text{  }2\text{  }}{\{z_{1}^{2}\}}$};
\path[->, thin] 
(q_1) 
edge [bend left] node {$\{z_{1}\}$} (q_2)
(q_2) 
edge [bend left] node {$\varnothing$} (q_1);
\end{tikzpicture}
} 
\subfigure[FSM $M_{2}$]{
\begin{tikzpicture}[->,
shorten >=0.1pt,%
auto,node distance=2.0cm,thick,
inner sep=0.1 pt ,bend angle=30]
\tikzstyle{every state}=[minimum size=8mm]
\tikzstyle{every node}=[font=\footnotesize]
\tikzstyle{initial}=[pin={[pin distance=5mm,pin edge={<-,shorten <=1pt}]left:$$}]
\node[state, thin, initial] (q_6) {$\frac{\text{  }3\text{  }}{\varnothing}$};
\node[state, thin] (q_7) [ right of=q_6] {$
\frac{\text{  }4\text{  }}{\{z_{2}^{2}\}}$};
\path[->, thin] 
(q_6) 
edge  [bend left] node {$\{z_{2}\}$} (q_7)
(q_7) 
edge  [bend left] node {$\varnothing$} (q_6);
\end{tikzpicture}
} \\
\subfigure[FSM $M_{3}$]{
\begin{tikzpicture}[->,
shorten >=0.1pt,%
auto,node distance=2.0cm,thick,
inner sep=0.1 pt ,bend angle=30]
\tikzstyle{every state}=[minimum size=8mm]
\tikzstyle{every node}=[font=\footnotesize]
\tikzstyle{initial}=[pin={[pin distance=5mm,pin edge={<-,shorten <=1pt}]left:$$}]
\node[state, thin, initial] (q_3) {$\frac{\text{  }5\text{  }}{\varnothing}$};
\node[state, thin] (q_4) [ below of=q_3] {$\frac{\text{  }6\text{  }}{\varnothing}$};
\node[state, thin] (q_5) [ right of=q_4] {$\frac{\text{  }7\text{  }}{\{\Vert z \Vert\}}$};
\path[->, thin] 
(q_3) 
edge node {$\varnothing$} (q_4)
(q_4) 
edge [bend left] node {$\{z_{1}^{2},z_{2}^{2}\}$} (q_5)
(q_5) 
edge [bend left] node {$\varnothing$} (q_4);
\end{tikzpicture}
}
\subfigure[AFSM $\mathbb{A}$]{
\begin{tikzpicture}[scale=1,->,
shorten >=0.1pt,
auto,node distance=2.0cm,thick,
inner sep=0.1 pt ,bend angle=30]
\tikzstyle{every state}=[minimum size=8mm,shape=rectangle]
\node[state, thin] (q_1) {$M_1$};
\node[state, thin] (q_2) [below of=q_1] {$M_2$};
\node[state, thin] (q_3) [ right of=q_2] {$M_3$};
\path[->, thin] 
(q_1) 
edge node [swap] {} (q_3)
(q_2) 
edge node [swap] {} (q_3);
\end{tikzpicture}
}  
\caption{AFSM $\mathbb{A}$ in Example \ref{exex}. }
\label{FSMs1}
\end{figure}
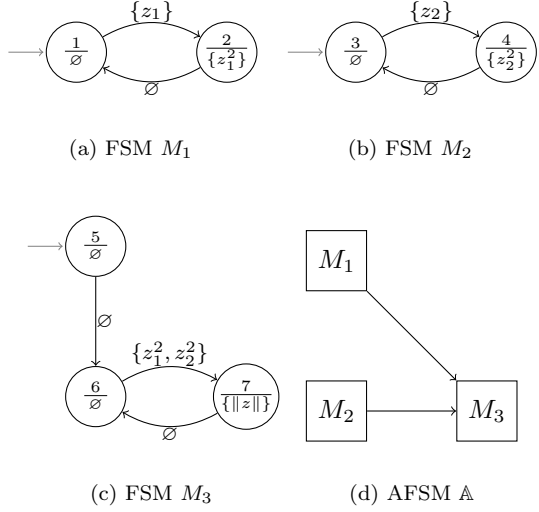
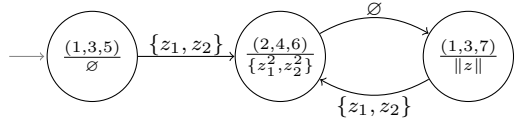
\begin{figure}[t]
\centering
\begin{tikzpicture}[->,
shorten >=0.1pt,
auto,node distance=2.5cm,thick,
inner sep=0.1 pt ,bend angle=30]
\tikzstyle{every state}=[minimum size=12mm]
\tikzstyle{every node}=[font=\footnotesize]
\tikzstyle{initial}=[pin={[pin distance=5mm,pin edge={<-,shorten <=1pt}]left:$$}]
\node[state, thin, initial] (q_1) {$\frac{(1,3,5)}{\varnothing}$};
\node[state, thin] (q_2) [ right of=q_1] {$\frac{(2,4,6)}{\{z_{1}^{2},z_{2}^{2}\}}$};
\node[state, thin] (q_3) [ right of=q_2] {$\frac{(1,3,7)}{\Vert z \Vert }$};
\path[->, thin] 
(q_1) 
edge node {$\{z_{1},z_{2}\}$} (q_2)
(q_2) 
edge [bend left] node {$\varnothing$} (q_3)
(q_3) 
edge [bend left] node {$\{z_{1},z_{2}\}$} (q_2);
\end{tikzpicture}
\caption{FSM $\mathbb{M}(\mathbb{A})$ in Example \ref{exex}.}
\label{M(A)}
\end{figure} 
\section{Compositional bisimulation of AFSMs}
A na\"ive approach to check bisimulation equivalence of two AFSMs $\mathbb{A}^{1}$ and $\mathbb{A}^{2}$ consists in first expanding them to FSMs $\mathbb{M}(\mathbb{A}^{1})$ and $\mathbb{M}(\mathbb{A}^{2})$ and then apply standard bisimulation algorithms (see e.g. \cite{BisAlg,Piazza,Hopcroft}). 
The main practical limitation of this approach resides in the well--known state explosion problem, see e.g. \cite{StateExplosion2,StateExplosion1}.
%, which refers to the exponential growth of AFSMs states' sets with the number of FSMs involved.  
This is the key reason for us to propose an alternative approach to check bisimulation equivalence of AFSMs %which does not need to expand AFSMs and
which is centered on the notion of \textit{compositional bisimulation} that is introduced hereafter.
\begin{definition}
\label{Hbis}
Given a pair of Arenas 
$\mathbb{A}^{j}=(\mathbb{V}^{j},\mathbb{E}^{j})$ of FSMs $M_{1}^{j}$, $M_{2}^{j}$, ..., $M_{N_{j}}^{j}$ ($j=1,2$), a set 
\mbox{$\mathbb{R}\subseteq \mathbb{V}^{1}\times \mathbb{V}^{2}$}, is a compositional bisimulation relation between $\mathbb{A}^{1}$ and $\mathbb{A}^{2}$ if for any $(M_{i}^{1},M_{j}^{2})\in \mathbb{R}$ the following conditions are satisfied:
\begin{itemize}
\item $M_{i}^{1}\cong M_{j}^{2}$;
\item existence of $(M_{i}^{1},M_{i^{\prime}}^{1})\in\mathbb{E}^{1}$ implies existence of $(M_{j}^{2},M_{j^{\prime}}^{2})\in\mathbb{E}^{2}$ such that $(M_{i^{\prime}}^{1},M_{j^{\prime}}^{2})\in \mathbb{R}$;
\item existence of $(M_{j}^{1},M_{j^{\prime}}^{2})\in\mathbb{E}^{2}$ implies existence of $(M_{i}^{1},M_{i^{\prime}}^{1})\in\mathbb{E}^{1}$ such that $(M_{i^{\prime}}^{1},M_{j^{\prime}}^{2})\in \mathbb{R}$.
\end{itemize}
The AFSMs $\mathbb{A}^{1}$ and $\mathbb{A}^{2}$ are compositionally bisimilar, denoted $\mathbb{A}^{1} \cong_{c} \mathbb{A}^{2}$, if there exists a total compositional bisimulation relation between $\mathbb{A}^{1}$ and $\mathbb{A}^{2}$.
\end{definition}
%It is readily seen that basic facts on bisimulation equivalence of FSMs recalled in Section \ref{Sec:FSM} can be adapted to compositional bisimulation of AFSMs, as follows. 
The notion of compositional bisimulation is an equivalence relation on the class of AFSMs. The maximal compositional bisimulation relation between AFSMs $\mathbb{A}^{1}$ and $\mathbb{A}^{2}$ is a compositional bisimulation relation $\mathbb{R}^{\ast}(\mathbb{A}^{1},\mathbb{A}^{2})$ such that $\mathbb{R}\subseteq \mathbb{R}^{\ast}(\mathbb{A}^{1},\mathbb{A}^{2})$ for any compositional bisimulation relation $\mathbb{R}$. The maximal compositional bisimulation exists and is unique. 
The set $\mathbb{R}^{\ast}(\mathbb{A},\mathbb{A})$ is an equivalence relation on the collection of FSMs in $\mathbb{A}$. 
The quotient
%\footnote{In the next section we show how quotients of AFSMs induced by compositional bisimulation can be computed as quotients of appropriate FSMs induced by ordinary bisimulation (Proposition \ref{gio}).}
 of $\mathbb{A}$ induced by $\mathbb{R}^{\ast}(\mathbb{A},\mathbb{A})$ is the minimal (in terms of the number of the FSMs involved) compositionally bisimilar AFSM of $\mathbb{A}$. The minimal AFSM of an AFSM $\mathbb{A}$, denoted $\mathbf{A}_{\min}(\mathbb{A})$, exists and is unique, up to isomorphisms. \\
Checking compositional bisimulation equivalence of AFSMs is equivalent to checking bisimulation equivalence of appropriate FSMs, as discussed hereafter. Consider a pair of AFSMs $\mathbb{A}^{j}=(\mathbb{V}^{j},\mathbb{E}^{j})$ ($j=1,2$). 
Since bisimulation is an equivalence relation on the set $\mathbb{V}_{1}\cup\mathbb{V}_{2}$ of FSMs, it induces a partition of $\mathbb{V}_{1}\cup\mathbb{V}_{2}$ in $K$ equivalence classes $C_{1},C_{2},...,C_{K}$ where $M_{i},M_{j}\in C_{k}$ if and only if $M_{i} \cong M_{j}$. Note that $\{C_{k}\}_{k\in K}$ is a finite set. Define the tuple: 
\begin{equation}
\label{FSMarena}
M_{\mathbb{A}^{j}}=(X_{\mathbb{A}^{j}},U_{\mathbb{A}^{j}},Y_{\mathbb{A}^{j}},H_{\mathbb{A}^{j}},\Delta_{\mathbb{A}^{j}}),
\end{equation}
where $X_{\mathbb{A}^{j}}=\mathbb{V}^{j}$, $U_{\mathbb{A}^{j}}=\varnothing$, $Y_{\mathbb{A}^{j}}=\{C_{k}\}_{k\in K}$, $H_{\mathbb{A}^{j}}:X_{\mathbb{A}^{j}}\rightarrow 2^{Y_{\mathbb{A}^{j}}}$ is defined by $H_{\mathbb{A}^{j}}(M_{i}^{j})=\{C_k\}$ if $M_{i}^{j}\in C_{k}$, and $\Delta_{\mathbb{A}^{j}}\subseteq X_{\mathbb{A}^{j}} \times \varnothing \times X_{\mathbb{A}^{j}}$ is such that 
$M_{i}^{j} \rTo_{\Delta_{\mathbb{A}^{j}}}^{\varnothing} M_{i^{\prime}}^{j}$ when $(M_{i}^{j},M_{i^{\prime}}^{j})\in\mathbb{E}^{j}$. By definition of $H_{\mathbb{A}^{j}}$, $H_{\mathbb{A}^{j}}(M_{i}^{j})=H_{\mathbb{A}^{j^{\prime}}}(M_{i^{\prime}}^{j^{\prime}})$ if and only if $M_{i}^{j} \cong M_{i^{\prime}}^{j^{\prime}}$.
%\begin{equation}
%\label{FSMarena}
%M_{\mathbb{A}^{j}}=(X_{\mathbb{A}^{j}},U_{\mathbb{A}^{j}},Y_{\mathbb{A}^{j}},H_{\mathbb{A}^{j}},\Delta_{\mathbb{A}^{j}}),
%\end{equation}
%where:
%\begin{itemize}
%\item $X_{\mathbb{A}^{j}}=\mathbb{V}^{j}$;
%\item $U_{\mathbb{A}^{j}}=\varnothing$;
%\item $Y_{\mathbb{A}^{j}}=\mathbb{V}^{1}\cup \mathbb{V}^{2}$;
%\item $H_{\mathbb{A}^{j}}:X_{\mathbb{A}^{j}}\rightarrow Y_{\mathbb{A}^{j}}$ is such that for any $M_{i}^{j},M_{i^{\prime}}^{j^{\prime}}\in \mathbb{V}^{1}\cup \mathbb{V}^{2}$, 
%$H_{\mathbb{A}^{j}}(M_{i}^{j})=H_{\mathbb{A}^{j^{\prime}}}(M_{i^{\prime}}^{j^{\prime}})$ if and only if $M_{i}^{j} \cong M_{i^{\prime}}^{j^{\prime}}$;
%\item $\Delta_{\mathbb{A}^{j}}\subseteq X_{\mathbb{A}^{j}} \times \{\varnothing\} \times X_{\mathbb{A}^{j}}$, such that 
%$M_{i}^{j} \rTo_{\Delta_{\mathbb{A}^{j}}}^{\varnothing} M_{i^{\prime}}^{j}$, if $(M_{i}^{j},M_{i^{\prime}}^{j})\in\mathbb{E}^{j}$.
%\end{itemize}
The syntax of the tuple in (\ref{FSMarena}) is the same as the one of FSMs from which, the following result holds.
\begin{proposition}
\label{gio}
%Consider a pair of AFSMs $\mathbb{A}^{1}$ and $\mathbb{A}^{2}$. Then 
$\mathbb{A}^{1} \cong_{c} \mathbb{A}^{2}$ if and only if $M_{\mathbb{A}^{1}} \cong M_{\mathbb{A}^{2}}$. 
\end{proposition}

\begin{proof}
By Definitions \ref{Bis} and \ref{Hbis}, it is readily seen that $\mathbb{A}^{1} \cong_{c} \mathbb{A}^{2}$ if and only if the set 
$\mathbb{R}^{\ast}(\mathbb{A}^{1},\mathbb{A}^{2})$ is a total bisimulation relation between $M_{\mathbb{A}^{1}}$ and $M_{\mathbb{A}^{2}}$. 
\end{proof}

%The proof of the above result is a direct consequence of the definition of $M_{\mathbb{A}^{j}}$, Definition \ref{Bis} and Definition \ref{Hbis}; it is therefore omitted. 
We are now ready to present the main result of the paper, that shows that the notion of compositional bisimulation of AFSMs is consistent with the notion of bisimulation of the corresponding expanded FSMs. 
\begin{theorem}\label{prop}
If $\mathbb{A}^{1} \cong_{c} \mathbb{A}^{2}$ then $\mathbb{M}(\mathbb{A}^{1}) \cong \mathbb{M}(\mathbb{A}^{2})$. 
\end{theorem}
\begin{proof}
Let be $\mathbb{A}^{j}=(\mathbb{V}^{j},\mathbb{E}^{j})$ with $\mathbb{V}^{j}=\{M_{1}^{j},M_{2}^{j},...,$ $M_{N_{j}}^{j}\}$ and 
$M_{i}^{j}=(X_{i}^{j},x_{i}^{0,j},U_{i}^{j},Y_{i}^{j},H_{i}^{j},\Delta_{i}^{j})$ ($i=1,2,...,N_{j}$, $j=1,2$). 
%We denote the FSMS obtained by expanding the corresponding AFSMs $\mathbb{A}^{j}$ by 
Set 
$ 
%\[
\mathbb{M}(\mathbb{A}^{j})=(X^{j},x^{0,j},U^{j},Y^{j},H^{j},$ $\Delta^{j}) 
%\]
$ ($j=1,2$). Since $\mathbb{A}^{1} \cong_{c} \mathbb{A}^{2}$, relation $\mathbb{R}^{\ast}(\mathbb{A}^{1},\mathbb{A}^{2})$ is total. 
Consider the relation $R\subseteq X^{1}\times X^{2}$ defined by $(x^{1},x^{2})\in R$ with $x^{1}=(x_{1}^{1},x_{2}^{1},...,x_{N_{1}}^{1})$ and $x^{2}=(x^{2}_{1},x^{2}_{2},...,x^{2}_{N_{2}})$ if and only if $(x_{i}^{1},x_{j}^{2})\in R^{\ast}(M_{i}^{1},M_{j}^{2})$ and $(M_{i}^{1},M^{2}_{j})\in \mathbb{R}^{\ast}(\mathbb{A}^{1},\mathbb{A}^{2})$. Consider %a pair of states 
$
%\[
(x^{1},x^{2})=((x_{1}^{1},x_{2}^{1},...,x_{N_{1}}^{1}),(x^{2}_{1},x^{2}_{2},...,x^{2}_{N_{2}}))\in R 
%\]
$. 
By definition of $R$, $H_{i}^{1}(x^{1}_{i})=H_{j}^{2}(x^{2}_{j})=\bigcup_{k\in I(i)}H_{k}^{2}(x^{2}_{k})$ for any $i=1,2,...,N_{1}$ and $j\in I(i)=\{k\in \{1,2,...,N_{2}\}\,|\, (M_{i}^{1},M^{2}_{k})\in \mathbb{R}^{\ast}(\mathbb{A}^{1},\mathbb{A}^{2})\}$. Hence one gets:
\begin{equation}
\label{qa1}
\begin{array}
{rcl}
H^{1}(x^{1}) & = & \bigcup_{M_{i}\in \mathbb{V}^{1}}H^{1}_{i}(x^{1}_{i})=\bigcup_{M_{i}\in \mathbb{V}^{1}}(\bigcup_{k\in I(i)}H^{2}_{k}(x^{2}_{k}))\nonumber\\
& = & \bigcup_{M_{j}\in \mathbb{V}^{2}} H^{2}_{j}(x^{2}_{j})=H^{2}(x^{2}).
\end{array}
\end{equation}
Note that the third equality in the above chain of equalities holds because $\mathbb{R}^{\ast}(\mathbb{A}^{1},\mathbb{A}^{2})$ is total.
Hence, condition (i) in Definition \ref{Bis} is satisfied. Consider any transition $
%\[
(x_{1}^{1},x_{2}^{1},...,x_{N_{1}}^{1})\rTo^{u^{1}} (z_{1}^{1},z_{2}^{1},...,z_{N_{1}}^{1})
%\]
$ in $\mathbb{M}(\mathbb{A}^{1})$. By definition of $\mathbb{M}(\mathbb{A}^{1})$ there exist transitions $x_{i}^{1}\rTo_{\Delta_{i}^{1}}^{u_{i}^{1}} z_{i}^{1}$ of $M_{i}^{1}$ ($i=1,2,...,N_{1}$), such that:
\begin{equation}
u^{1}=\bigcup_{M_{i}\in \mathbb{V}^{1}} (u^{1}_{i} \backslash (\bigcup_{M_{i^{\prime}}\in \Pre(\mathbb{A}^{1},M_{i})}H_{i^{\prime}}^{1}(x_{i^{\prime}}^{1}))).
\label{fff}
\end{equation}
By definition of $R$, for any $i=1,2,...,N_{1}$ there exist transitions $x_{j}^{2}\rTo_{\Delta_{j}^{2}}^{u_{j}^{2}} z_{j}^{2}$ of $M_{j}^{2}$, with $(M_{i}^{1},M_{j}^{2})\in\mathbb{R}^{\ast}(\mathbb{A}^{1},\mathbb{A}^{2})$ and
\begin{eqnarray}
&& (z_{i}^{1},z_{j}^{2})\in R^{\ast}(M_{i}^{1},M_{j}^{2}),\label{aaa}\\
&& u_{i}^{1}=u_{j}^{2},\label{ccc}\\
&& H_{i}^{1}(x_{i}^{1})=H_{j}^{2}(x_{j}^{2}).\label{cccc}
\end{eqnarray}
Set:
\begin{equation}
u^{2}=\bigcup_{M_{j}\in \mathbb{V}^{2}} (u^{2}_{j} \backslash (\bigcup_{M_{j^{\prime}}\in \Pre(\mathbb{A}^{2},M_{j})}H_{j^{\prime}}^{2}(x_{j^{\prime}}^{2})),
\label{ffff}
\end{equation}
and consider the transition $
%\begin{equation}
(x_{1}^{2},x_{2}^{2},...,x_{N_{2}}^{2})\rTo^{u^{2}}$ $(z_{1}^{2},z_{2}^{2},...,z_{N_{2}}^{2})%.\nonumber
%\label{TrAnS}
%\end{equation}
$ in $\mathbb{M}(\mathbb{A}^{2})$. 
%By construction, the above transition is in $\mathbb{M}(\mathbb{A}^{2})$. 
By definition of the relation $R$ and by combining (\ref{fff}), (\ref{ccc}), (\ref{cccc}) and (\ref{ffff}), one gets $u^{1}=u^{2}$ from which, together with condition (\ref{aaa}), one gets $
%\[
((z_{1}^{1},z_{2}^{1},...,z_{N_{1}}^{1}),(z^{2}_{1},z^{2}_{2},...,z^{2}_{N_{2}}))\in R
%\]
$. 
Thus, condition (ii) in Definition \ref{Bis} is proved. Condition (iii) can be proven by using similar arguments. Finally condition (iv) holds by definition of $R$. 
\end{proof}

%The above result is important because it provides a method to assess bisimulation equivalence of AFSMs $\mathbb{A}^{i}$ without expanding them to the corresponding FSMs $\mathbb{M}(\mathbb{A}^{i})$. 
The converse implication of the above result, i.e. whether $\mathbb{M}(\mathbb{A}^{1}) \cong \mathbb{M}(\mathbb{A}^{2})$ implies $\mathbb{A}^{1} \cong_{c} \mathbb{A}^{2}$, is not true in general, as shown in the following counterexample. 
\begin{example}
Consider four FSMs $M_{i}=(X_{i},X^{0}_{i},U_{i},$ $Y_{i},H_{i},\Delta_{i})$, where each $M_{i}$ is characterized by the unique transition $x_{i}^{0}\rTo_{\Delta_{i}}^{u_{i}}x_{i}^{+}$, where:
\[
\begin{tabular}
[c]{l|cccc}
\hline
 & $M_{1}$ & $M_{2}$ & $M_{3}$ & $M_{4}$\\
\hline
$u_{i}$ & $\{a\}$ & $\{c\}$ & $\{b,d\}$ & $\{a,d\}$ \\
$H_{i}(x_{i}^{0})$ & $\{b\}$ & $\{d\}$ & $\{e\}$ & $\{b,e\}$\\
$H_{i}(x_{i}^{+})$ & $\{f\}$ & $\{f\}$ & $\{f\}$ & $\{f\}$\\
\hline
\end{tabular}
\]
Consider a pair of AFSMs $\mathbb{A}^{1}=(\mathbb{V}^{1},\mathbb{E}^{1})$ and $\mathbb{A}^{2}=(\mathbb{V}^{2},\mathbb{E}^{2})$, depicted in Figure \ref{figcounterexample}, 
\begin{figure}[t]
\centering
\begin{tikzpicture}[scale=1,->,
shorten >=0.1pt,%
auto,node distance=1.65cm,thick,
inner sep=0.1 pt ,bend angle=30]
\tikzstyle{every state}=[minimum size=6mm,shape=rectangle]
\node[state, thin] (q_1) {$M_1$};
\node[state, thin] (q_3) [ right of=q_1] {$M_3$};
\node[state, thin] (q_2) [ right of=q_3] {$M_2$};
\node[state, thin] (q_4) [ right of=q_2] {$M_2$};
\node[state, thin] (q_5) [ right of=q_4] {$M_4$};
\path[->, thin] 
(q_1) 
edge node [swap] {} (q_3)
(q_2) 
edge node {} (q_3)
(q_4) 
edge node {} (q_5);
\end{tikzpicture}
\caption[]{AFSM $\mathbb{A}_{1}$ in the left and AFSM $\mathbb{A}_{2}$ in the right.}
\label{figcounterexample}
\end{figure}
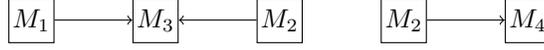
where:
\[
\begin{array}
{ll}
\mathbb{V}^{1}=\{M_{1},M_{2},M_{3}\}, & \mathbb{E}^{1}=\{(M_{1},M_{3}),(M_{2},M_{3})\},\\
\mathbb{V}^{2}=\{M_{2},M_{4}\}, & \mathbb{E}^{2}=\{(M_{2},M_{4})\}. 
\end{array}
\]
% $\mathbb{V}^{1}=\{M_{1},M_{2},M_{3}\}$, $\mathbb{E}^{1}=\{(M_{1},M_{3}),$ $(M_{2},M_{3})\}$, $\mathbb{V}^{2}=\{M_{2},M_{4}\}$ and $\mathbb{E}^{2}=\{(M_{2},M_{4})\}$. 
It is easy to see that FSM $\mathbb{M}(\mathbb{A}^{1})$ is composed by the unique transition:
\[
(x_{1}^{0},x_{2}^{0},x_{3}^{0}) \rTo^{\{a,c\}} (x_{1}^{+},x_{2}^{+},x_{3}^{+}),
\]
with output function $H^{1}$ defined by $H^{1}(x_{1}^{0},x_{2}^{0},x_{3}^{0})=\{b,d,e\}$ and $H^{1}(x_{1}^{+},x_{2}^{+},x_{3}^{+})=\{f\}$. Moreover, FSM $\mathbb{M}(\mathbb{A}^{2})$ is characterized by the unique transition:
\[
(x_{2}^{0},x_{4}^{0}) \rTo^{\{a,c\}} (x_{2}^{+},x_{4}^{+}),
\]
with output function $H^{2}$ defined by $H^{2}(x_{2}^{0},x_{4}^{0})=\{b,d,e\}$ and $H^{2}(x_{2}^{+},x_{4}^{+})=\{f\}$. Hence, FSMs $\mathbb{M}(\mathbb{A}^{1})$ and $\mathbb{M}(\mathbb{A}^{2})$ are bisimilar. On the other hand, it is easy to see that FSM $M_{4}$ is bisimilar to no FSM $M_{i}$, $i=1,2,3$. Hence, $\mathbb{A}^{1}$ and $\mathbb{A}^{2}$ are not compositionally bisimilar.
\end{example}

Theorem \ref{prop} can be used to reduce the size of AFSMs through compositional bisimulation, as follows. 
%Given $\mathbb{A}$, we recall that $\mathbf{M}_{\min}(\mathbb{M}(\mathbb{A}))$ denotes the minimal bisimilar FSM of $\mathbb{M}(\mathbb{A})$ and $\mathbf{A}_{\min}(\mathbb{A})$ denotes the minimal compositionally bisimilar AFSM of $\mathbb{A}$.
\begin{corollary}
\label{coro}
$\mathbf{M}_{\min}(\mathbb{M}(\mathbb{A})) \cong^{\iso} \mathbf{M}_{\min}(\mathbb{M}(\mathbf{A}_{\min}(\mathbb{A})))$.
\end{corollary}
\begin{proof}
By definition of $\mathbf{M}_{\min}$, $\mathbf{M}_{\min}(\mathbb{M}(\mathbb{A})) \cong \mathbb{M}(\mathbb{A})$ and $\mathbb{M}(\mathbf{A}_{\min}(\mathbb{A}))\cong \mathbf{M}_{\min}(\mathbb{M}(\mathbf{A}_{\min}(\mathbb{A})))$. Since $\mathbb{A}\cong_{c} \mathbf{A}_{\min}(\mathbb{A})$, by Theorem \ref{prop}, $\mathbb{M}(\mathbb{A}) \cong \mathbb{M}(\mathbf{A}_{\min}(\mathbb{A}))$.  
Hence,
\[
\mathbf{M}_{\min}(\mathbb{M}(\mathbb{A})) \cong \mathbb{M}(\mathbb{A}) \cong \mathbb{M}(\mathbf{A}_{\min}(\mathbb{A})) 
\cong \mathbf{M}_{\min}(\mathbb{M}(\mathbf{A}_{\min}(\mathbb{A})))
\]
that, by transitivity implies 
%Hence:
%\[
%\begin{array}
%{rcl}
%\mathbf{M}_{\min}(\mathbb{M}(\mathbb{A})) & \cong & \mathbb{M}(\mathbb{A}) \cong \mathbb{M}(\mathbf{A}_{\min}(\mathbb{A})) \nonumber\\
%& \cong & \mathbf{M}_{\min}(\mathbb{M}(\mathbf{A}_{\min}(\mathbb{A}))).
%\end{array}
%\]
%By the transitivity property of bisimulation equivalence, the above chain of bisimulation equivalences implies 
%\[
$\mathbf{M}_{\min}(\mathbb{M}(\mathbb{A})) \cong \mathbf{M}_{\min}(\mathbb{M}(\mathbf{A}_{\min}(\mathbb{A})))$ 
%\]
which, by Lemma \ref{lemma}, concludes the proof.
\end{proof}

The above result suggests a method to use compositional bisimulation for complexity reduction of AFSMs, as summarized in the following algorithm:
\begin{itemize}
\item Compute the relation $\mathbb{R}^{\ast}(\mathbb{A},\mathbb{A})$.
\item Compute the quotient $\mathbf{A}_{\min}(\mathbb{A})$.
\item Expand the AFSM $\mathbf{A}_{\min}(\mathbb{A})$ to the FSM $\mathbb{M}(\mathbf{A}_{\min}(\mathbb{A}))$.
\item Compute the relation $R^{\ast}(\mathbb{M}(\mathbf{A}_{\min}(\mathbb{A})),\mathbb{M}(\mathbf{A}_{\min}(\mathbb{A})))$.
\item Compute the quotient $\mathbf{M}_{\min}(\mathbb{M}(\mathbf{A}_{\min}(\mathbb{A})))$. % of $\mathbb{M}(\mathbf{A}_{\min}(\mathbb{A}))$.
\end{itemize}
%The benefits from the use of this algorithm are quantified in the next section.
%, through a computational complexity analysis, and illustrated in Section 6 through an example in the context of systems biology.

\section{Complexity analysis}

In this section we compare computational complexity in checking compositional bisimulation equivalence between AFSMs and bisimulation equivalence between the corresponding expanded FSMs. Consider a pair of AFSMs $\mathbb{A}^{i}=(\mathbb{V}^{i},\mathbb{E}^{i})$ composed of $N_{i}$ FSMs and set $\mathbb{M}(\mathbb{A}^{i})=(X^{i},x^{0}_{i},U^{i},Y^{i},H^{i},\Delta^{i})$ ($i=1,2$). As common practice in the analysis of non--flat systems, e.g. \cite{StateExplosion2,StateExplosion1}, in the sequel we evaluate how computational complexity scales with the number $N_{i}$ of FSMs in AFSMs $\mathbb{A}^{i}$. 
%We start by noting that $|X^{i}|\sim 2^{N_{i}}$, $|\Delta^{i}|\sim 2^{2N_{i}}\sim 2^{N_{i}}$ and that the inequality $\frac{1}{2^{N_{1}}}+\frac{1}{2^{N_{2}}}\leq 1$ (which is true for any $N_{1},N_{2}\in\mathbb{N}$) implies $2^{N_{1}}+2^{N_{2}}\leq 2^{N_{1}+N_{2}}$ from which, the following results hold as a direct application of Propositions \ref{prop21} and \ref{prop22}. 
We start by evaluating the computational complexity in checking bisimulation equivalence of the flattened systems $\mathbb{M}(\mathbb{A}^{1})$ and $\mathbb{M}(\mathbb{A}^{2})$. As a direct application of Propositions \ref{prop21} and \ref{prop22}, one gets the following results.

\begin{corollary}
\label{prop23}
Space complexity in checking $\mathbb{M}(\mathbb{A}^{1})\cong \mathbb{M}(\mathbb{A}^{2})$ is $O(2^{N_{1}}+2^{N_{2}})$.
\end{corollary}

\begin{corollary}
\label{prop24}
Time complexity in checking $\mathbb{M}(\mathbb{A}^{1})\cong \mathbb{M}(\mathbb{A}^{2})$ is $O((2^{N_{1}}+2^{N_{2}})\ln(2^{N_{1}}+2^{N_{2}}))$. 
\end{corollary}

The above result quantifies the aforementioned state explosion problem \cite{StateExplosion2,StateExplosion1} in the class of AFSMs. We now discuss computational complexity in checking compositional bisimulation. 
\begin{proposition}
\label{prop25}
Space complexity in checking $\mathbb{A}^{1}\cong_{c}\mathbb{A}^{2}$ is $O(N_{1}^{2}+N_{2}^{2})$.
\end{proposition}

\begin{proof}
Direct consequence of Propositions \ref{prop21} and \ref{gio}. 
\end{proof}

\begin{proposition}
\label{prop26}
Time complexity in checking $\mathbb{A}^{1}\cong_{c} \mathbb{A}^{2}$ is $O((N_{1}^{2}+N_{2}^{2})\ln(N_{1}+N_{2}))$. 
\end{proposition}
\begin{proof}
By Proposition \ref{gio}, checking $\mathbb{A}^{1}\cong_{c} \mathbb{A}^{2}$ reduces to: (1) construct FSMs $M_{\mathbb{A}^{1}}$ and $M_{\mathbb{A}^{2}}$ and, (2) check if $M_{\mathbb{A}^{1}}\cong M_{\mathbb{A}^{2}}$. 
Regarding (1), time complexity effort reduces to the one of defining functions $H_{\mathbb{A}^{1}}$ and $H_{\mathbb{A}^{2}}$ which amounts to 
$O((N_{1}+N_{2})^{2})$. 
Regarding (2), by Proposition \ref{prop22}, time complexity in checking $M_{\mathbb{A}^{1}} \cong M_{\mathbb{A}^{2}}$ is given by $
O((N_{1}^{2}+N_{2}^{2})\ln(N_{1}+N_{2}))$.
Since the last term is dominant over $O((N_{1}+N_{2})^{2})$, the result follows. 
\end{proof}

%By Corollaries \ref{prop23} and \ref{prop24} and Propositions \ref{prop25} and \ref{prop26} it is readily seen that the computational complexity in checking compositional bisimulation between AFSMs scales \textit{polynomially} with $N_{1}+N_{2}$, while the one in checking bisimulation between the corresponding FSMs scales \textit{exponentially} with $N_{1}+N_{2}$. 
%In \cite{AFSMrep} we applied AFSMs modeling paradigm to the regulation of gene expression of the single--celled bacterium \textit{E. coli}. An AFSM $\mathbb{A}$ has been presented which models the metabolism of glucose, galactose and arabinose. The expanded FSM $\mathbb{M}(\mathbb{A})$ is composed of $4,831,838,208$ states. By applying the results reported in this paper the minimal FSM $\mathbf{M}_{\min}(\mathbb{A})$ has been computed, resulting in $55,296$ states.

\section{Regulation of gene expression in \textit{E. coli}}
\begin{figure}[t] 
\label{EColi}
\begin{center}
\includegraphics[scale=0.32]{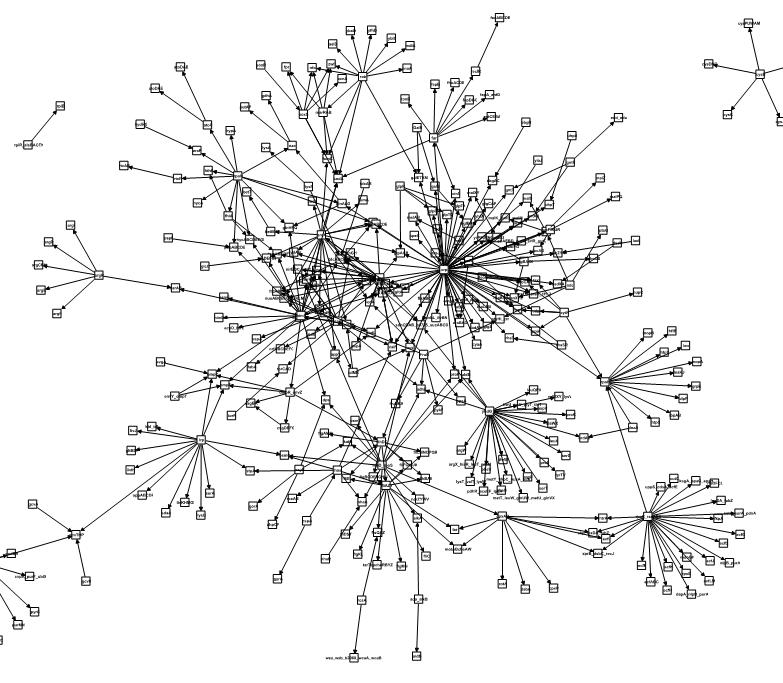}
\caption{A transcription network that represent about 10\% of the transcription interactions in the bacterium \textit{E. coli}.}
\end{center} 
\end{figure}
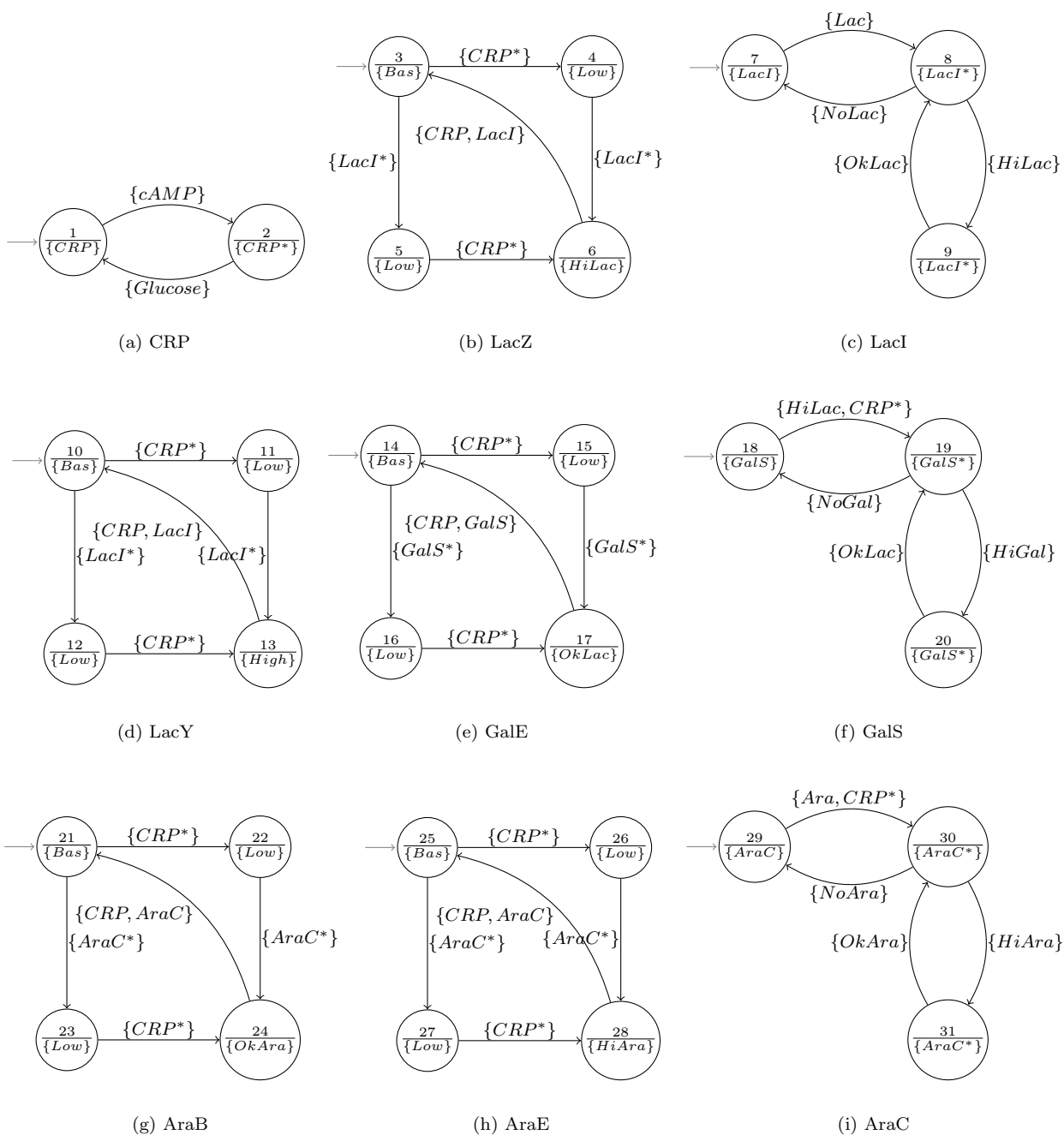
\begin{figure*}[t]
\centering
\subfigure[CRP]{
\begin{tikzpicture}[->,
shorten >=0.1pt,
auto,node distance=3cm,thick,
inner sep=0.1 pt ,bend angle=30]
\tikzstyle{every state}=[minimum size=8mm]
\tikzstyle{every node}=[font=\footnotesize]
\tikzstyle{initial}=[pin={[pin distance=5mm,pin edge={<-,shorten <=1pt}]left:$$}]
\node[state, thin, initial] (q_1) {$\frac{\text{  }1\text{  }}{\{CRP\}}$ };
\node[state, thin] (q_2) [ right of=q_1] {$\frac{\text{  }2\text{  }}{\{CRP^{\ast}\}}$};
\path[->, thin] 
(q_1) 
edge [bend left] node {$\{cAMP\}$} (q_2)
(q_2) 
edge [bend left] node {$\{Glucose\}$} (q_1);
\end{tikzpicture}
} 
\subfigure[LacZ]{
\begin{tikzpicture}[->,
shorten >=0.1pt,%
auto,node distance=3cm,thick,
inner sep=0.1 pt ,bend angle=30]
\tikzstyle{every state}=[minimum size=8mm]
\tikzstyle{every node}=[font=\footnotesize]
\tikzstyle{initial}=[pin={[pin distance=5mm,pin edge={<-,shorten <=1pt}]left:$$}]
\node[state, thin, initial] (q_6) {$\frac{\text{  }3\text{  }}{\{Bas\}}$};
\node[state, thin] (q_7) [ right of=q_6] {$
\frac{\text{  }4\text{  }}{\{Low\}}$};
\node[state, thin] (q_8) [below of=q_6] {$\frac{\text{  }5\text{  }}{\{Low\}}$};
\node[state, thin] (q_9) [ right of=q_8] {$\frac{\text{  }6\text{  }}{\{HiLac\}}$};
\path[->, thin] 
(q_6) 
edge  node {$\{CRP^{\ast}\}$} (q_7)
edge  node [swap] {$\{LacI^{\ast}\}$} (q_8)
(q_7) 
edge  node {$\{LacI^{\ast}\}$} (q_9)
(q_8) 
edge  node {$\{CRP^{\ast}\}$} (q_9)
(q_9) 
edge [bend right] node {$\{CRP,LacI\}$} (q_6);
\end{tikzpicture}
} 
\subfigure[LacI]{
\begin{tikzpicture}[->,
shorten >=0.1pt,%
auto,node distance=3cm,thick,
inner sep=0.1 pt ,bend angle=30]
\tikzstyle{every state}=[minimum size=8mm]
\tikzstyle{every node}=[font=\footnotesize]
\tikzstyle{initial}=[pin={[pin distance=5mm,pin edge={<-,shorten <=1pt}]left:$$}]
\node[state, thin, initial] (q_3) {$\frac{\text{  }7\text{  }}{\{LacI\}}$};
\node[state, thin] (q_4) [ right of=q_3] {$\frac{\text{  }8\text{  }}{\{LacI^{\ast}\}}$};
\node[state, thin] (q_5) [ below of=q_4] {$\frac{\text{  }9\text{  }}{\{LacI^{\ast}\}}$};
\path[->, thin] 
(q_3) 
edge [bend left] node {$\{Lac\}$} (q_4)
(q_4) 
edge [bend left] node {$\{NoLac\}$} (q_3)
edge [bend left] node {$\{HiLac\}$} (q_5)
(q_5) 
edge [bend left] node {$\{OkLac\}$} (q_4);
\end{tikzpicture}
}\\
\subfigure[LacY]{
\begin{tikzpicture}[->,
shorten >=0.1pt,%
auto,node distance=3cm,thick,
inner sep=0.1 pt ,bend angle=30]
\tikzstyle{every state}=[minimum size=8mm]
\tikzstyle{every node}=[font=\footnotesize]
\tikzstyle{initial}=[pin={[pin distance=5mm,pin edge={<-,shorten <=1pt}]left:$$}]
\node[state, thin, initial] (q_6) {$\frac{\text{  }10\text{  }}{\{Bas\}}$};
\node[state, thin] (q_7) [ right of=q_6] {$
\frac{\text{  }11\text{  }}{\{Low\}}$};
\node[state, thin] (q_8) [below of=q_6] {$\frac{\text{  }12\text{  }}{\{Low\}}$};
\node[state, thin] (q_9) [ right of=q_8] {$
\frac{\text{  }13\text{  }}{\{High\}}$};
\path[->, thin] 
(q_6) 
edge  node {$\{CRP^{\ast}\}$} (q_7)
edge  node {$\{LacI^{\ast}\}$} (q_8)
(q_7) 
edge  node [swap] {$\{LacI^{\ast}\}$} (q_9)
(q_8) 
edge  node {$\{CRP^{\ast}\}$} (q_9)
(q_9) 
edge [bend right] node {$\{CRP,LacI\}$} (q_6);
\end{tikzpicture}
} 
\subfigure[GalE]{
\begin{tikzpicture}[->,
shorten >=0.1pt,%
auto,node distance=3cm,thick,
inner sep=0.1 pt ,bend angle=30]
\tikzstyle{every state}=[minimum size=8mm]
\tikzstyle{every node}=[font=\footnotesize]
\tikzstyle{initial}=[pin={[pin distance=5mm,pin edge={<-,shorten <=1pt}]left:$$}]
\node[state, thin, initial] (q_6) {$\frac{\text{  }14\text{  }}{\{Bas\}}$};
\node[state, thin] (q_7) [ right of=q_6] {$
\frac{\text{  }15\text{  }}{\{Low\}}$};
\node[state, thin] (q_8) [below of=q_6] {$\frac{\text{  }16\text{  }}{\{Low\}}$};
\node[state, thin] (q_9) [ right of=q_8] {$
\frac{\text{  }17\text{  }}{\{OkLac\}}$};
\path[->, thin] 
(q_6) 
edge  node {$\{CRP^{\ast}\}$} (q_7)
edge  node {$\{GalS^{\ast}\}$} (q_8)
(q_7) 
edge  node {$\{GalS^{\ast}\}$} (q_9)
(q_8) 
edge  node {$\{CRP^{\ast}\}$} (q_9)
(q_9) 
edge [bend right] node {$\{CRP,GalS\}$} (q_6);
\end{tikzpicture}
}
\subfigure[GalS]{
\begin{tikzpicture}[->,
shorten >=0.1pt,%
auto,node distance=3cm,thick,
inner sep=0.1 pt ,bend angle=30]
\tikzstyle{every state}=[minimum size=8mm]
\tikzstyle{every node}=[font=\footnotesize]
\tikzstyle{initial}=[pin={[pin distance=5mm,pin edge={<-,shorten <=1pt}]left:$$}]
\node[state, thin, initial] (q_3) {$\frac{\text{  }18\text{  }}{\{GalS\}}$};
\node[state, thin] (q_4) [ right of=q_3] {$\frac{\text{  }19\text{  }}{\{GalS^{\ast}\}}$};
\node[state, thin] (q_5) [ below of=q_4] {$\frac{\text{  }20\text{  }}{\{GalS^{\ast}\}}$};
\path[->, thin] 
(q_3) 
edge [bend left] node {$\{HiLac,CRP^{\ast}\}$} (q_4)
(q_4) 
edge [bend left] node {$\{NoGal\}$} (q_3)
edge [bend left] node {$\{HiGal\}$} (q_5)
(q_5) 
edge [bend left] node {$\{OkLac\}$} (q_4);
\end{tikzpicture}
}\\
\subfigure[AraB]{
\begin{tikzpicture}[->,
shorten >=0.1pt,%
auto,node distance=3cm,thick,
inner sep=0.1 pt ,bend angle=30]
\tikzstyle{every state}=[minimum size=8mm]
\tikzstyle{every node}=[font=\footnotesize]
\tikzstyle{initial}=[pin={[pin distance=5mm,pin edge={<-,shorten <=1pt}]left:$$}]
\node[state, thin, initial] (q_6) {$\frac{\text{  }21\text{  }}{\{Bas\}}$};
\node[state, thin] (q_7) [ right of=q_6] {$
\frac{\text{  }22\text{  }}{\{Low\}}$};
\node[state, thin] (q_8) [below of=q_6] {$\frac{\text{  }23\text{  }}{\{Low\}}$};
\node[state, thin] (q_9) [ right of=q_8] {$
\frac{\text{  }24\text{  }}{\{OkAra\}}$};
\path[->, thin] 
(q_6) 
edge  node {$\{CRP^{\ast}\}$} (q_7)
edge  node {$\{AraC^{\ast}\}$} (q_8)
(q_7) 
edge  node {$\{AraC^{\ast}\}$} (q_9)
(q_8) 
edge  node {$\{CRP^{\ast}\}$} (q_9)
(q_9) 
edge [bend right] node {$\{CRP,AraC\}$} (q_6);
%edge [bend left] node [swap] {$\{CRP,AraC\}$} (q_6);
\end{tikzpicture}
} 
\subfigure[AraE]{
\begin{tikzpicture}[->,
shorten >=0.1pt,%
auto,node distance=3cm,thick,
inner sep=0.1 pt ,bend angle=30]
\tikzstyle{every state}=[minimum size=8mm]
\tikzstyle{every node}=[font=\footnotesize]
\tikzstyle{initial}=[pin={[pin distance=5mm,pin edge={<-,shorten <=1pt}]left:$$}]
\node[state, thin, initial] (q_6) {$\frac{\text{  }25\text{  }}{\{Bas\}}$};
\node[state, thin] (q_7) [ right of=q_6] {$
\frac{\text{  }26\text{  }}{\{Low\}}$};
\node[state, thin] (q_8) [below of=q_6] {$\frac{\text{  }27\text{  }}{\{Low\}}$};
\node[state, thin] (q_9) [ right of=q_8] {$
\frac{\text{  }28\text{  }}{\{HiAra\}}$};
\path[->, thin] 
(q_6) 
edge  node {$\{CRP^{\ast}\}$} (q_7)
edge  node {$\{AraC^{\ast}\}$} (q_8)
(q_7) 
edge  node [swap] {$\{AraC^{\ast}\}$} (q_9)
(q_8) 
edge  node {$\{CRP^{\ast}\}$} (q_9)
(q_9) 
edge [bend right] node {$\{CRP,AraC\}$} (q_6);
\end{tikzpicture}
}
\subfigure[AraC]{
\begin{tikzpicture}[->,
shorten >=0.1pt,%
auto,node distance=3cm,thick,
inner sep=0.1 pt ,bend angle=30]
\tikzstyle{every state}=[minimum size=8mm]
\tikzstyle{every node}=[font=\footnotesize]
\tikzstyle{initial}=[pin={[pin distance=5mm,pin edge={<-,shorten <=1pt}]left:$$}]
\node[state, thin, initial] (q_3) {$\frac{\text{  }29\text{  }}{\{AraC\}}$};
\node[state, thin] (q_4) [ right of=q_3] {$\frac{\text{  }30\text{  }}{\{AraC^{\ast}\}}$};
\node[state, thin] (q_5) [ below of=q_4] {$\frac{\text{  }31\text{  }}{\{AraC^{\ast}\}}$};
\path[->, thin] 
(q_3) 
edge [bend left] node {$\{Ara,CRP^{\ast}\}$} (q_4)
(q_4) 
edge [bend left] node {$\{NoAra\}$} (q_3)
edge [bend left] node {$\{HiAra\}$} (q_5)
(q_5) 
edge [bend left] node {$\{OkAra\}$} (q_4);
\end{tikzpicture}
}
\caption{FSMs modeling proteins involved in the AFSM $\mathbb{A}$.}
\label{FSMss}
\end{figure*}
\begin{figure*}[t]
\label{BioAFSM}
\centering
\begin{tikzpicture}[scale=1,->,shorten >=0.1pt,auto,node distance=2.0cm,thick,inner sep=0.1 pt ,bend angle=30]
\tikzstyle{every state}=[minimum size=10mm,shape=rectangle]
\node[state, thin] (q_1) {LacZ};
\node[state, thin] (q_2) [right of=q_1] {LacY};
\node[state, thin] (q_10) [right of=q_2] {LacA};
\node[state, thin] (q_3) [right of=q_10] {AraE};
\node[state, thin] (q_11) [right of=q_3] {AraF};
\node[state, thin] (q_12) [right of=q_11] {AraG};
\node[state, thin] (q_4) [below of=q_1] {LacI};
\node[state, thin] (q_13) [right of=q_4] {GalT};
\node[state, thin] (q_5) [right of=q_13] {CPR};
\node[state, thin] (q_6) [below of=q_11] {AraC};
\node[state, thin] (q_14) [right of=q_6] {AraH};
\node[state, thin] (q_7) [below of=q_4] {GalS};
\node[state, thin] (q_8) [right of=q_7] {GalE};
\node[state, thin] (q_15) [right of=q_8] {GalK};
\node[state, thin] (q_9) [ right of=q_15] {AraB};
\node[state, thin] (q_16) [right of=q_9] {AraA};
\node[state, thin] (q_17) [ right of=q_16] {AraD};
\node[state, thin] (q_91) [right of=q_12] {LacZ};
\node[state, thin] (q_92) [right of=q_91] {LacY};
\node[state, thin] (q_93) [ right of=q_92] {AraE};
\node[state, thin] (q_94) [below of=q_91] {LacI};
\node[state, thin] (q_95) [right of=q_94] {CPR};
\node[state, thin] (q_96) [ right of=q_95] {AraC};
\node[state, thin] (q_97) [below of=q_94] {GalS};
\node[state, thin] (q_98) [right of=q_97] {GalE};
\node[state, thin] (q_99) [ right of=q_98] {AraB};
\path[->, thin] 
(q_1) 
edge [bend right] node {} (q_7)
edge [bend right] node {} (q_4)
(q_3) 
edge [bend right] node {} (q_6)
(q_4) 
edge [bend right] node {} (q_1)
edge node {} (q_2)
edge node {} (q_10)
(q_5) 
edge node {} (q_1)
edge node {} (q_2)
edge node {} (q_3)
edge node {} (q_6)
edge node {} (q_7)
edge node {} (q_9)
edge node {} (q_8)
edge node {} (q_10)
edge node {} (q_11)
edge node {} (q_12)
edge node {} (q_13)
edge node {} (q_15)
edge [bend left] node {} (q_14)
edge node {} (q_16)
edge node {} (q_17)
(q_6) 
edge [bend right] node {} (q_3)
edge [bend right] node {} (q_9)
edge [bend right] node {} (q_11)
edge [bend right] node {} (q_12)
edge [bend right] node {} (q_14)
edge [bend right] node {} (q_16)
edge [bend right] node {} (q_17)
(q_7) 
edge [bend right] node {} (q_8)
edge [bend right] node {} (q_13)
edge [bend right] node {} (q_15)
(q_8) 
edge [bend right] node {} (q_7)
edge node {} (q_4)
(q_9) 
edge [bend right] node {} (q_6)
(q_11) 
edge [bend right] node {} (q_6)
(q_12) 
edge [bend right] node {} (q_6)
(q_13) 
edge [bend right] node {} (q_7)
(q_14) 
edge [bend right] node {} (q_6)
(q_15) 
edge [bend right] node {} (q_7)
(q_16) 
edge [bend right] node {} (q_6)
(q_17) 
edge [bend right] node {} (q_6)
(q_91) 
edge [bend right] node {} (q_97)
edge [bend right] node {} (q_94)
(q_93) 
edge [bend right] node {} (q_96)
(q_94) 
edge [bend right] node {} (q_91)
edge node {} (q_92)
(q_95) 
edge node {} (q_91)
edge node {} (q_92)
edge node {} (q_93)
edge node {} (q_96)
edge node {} (q_97)
edge node {} (q_99)
edge node {} (q_98)
(q_96) 
edge [bend right] node {} (q_93)
edge [bend right] node {} (q_99)
(q_97) 
edge [bend right] node {} (q_98)
(q_98) 
edge [bend right] node {} (q_97)
edge node {} (q_94)
(q_99) 
edge [bend right] node {} (q_96);
\end{tikzpicture}
\caption{AFSM $\mathbb{A}$ in the left panel and the minimal AFSM $\mathbf{A}_{\min}(\mathbb{A})$ in the right panel.}
\end{figure*}
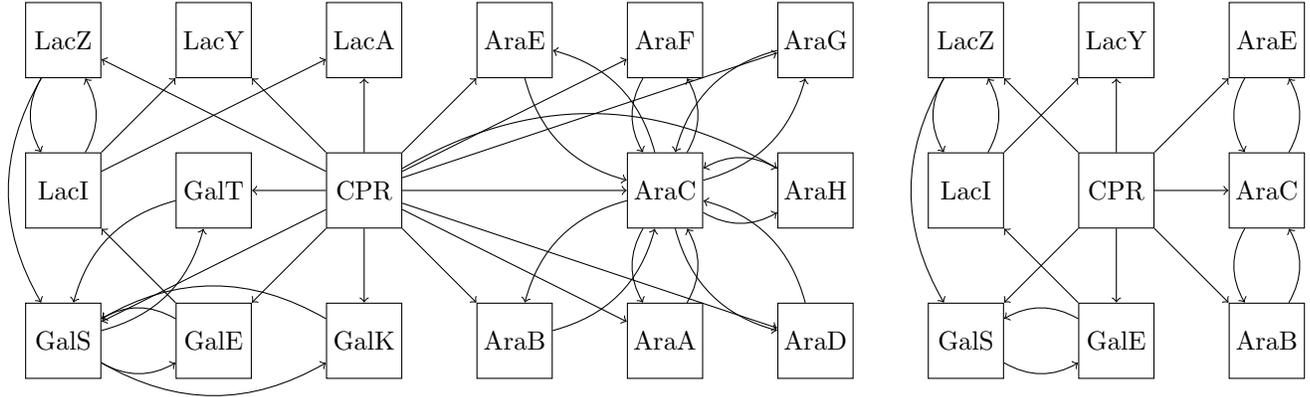
Several mathematical models have been proposed in the control systems' and computer science literature to model genetic regulatory systems, see e.g. \cite{BioSurvey}, and the references therein. We recall directed and undirected graphs, bayesan, boolean and generalized logical networks in the class of discrete systems, and nonlinear, piecewise--linear, qualitative, partial differential equations in the class of continuous systems. Stochastic hybrid systems have been proposed in \cite{Julius:08}. These models can be broadly classified along two orthogonal mathematical paradigms: (i) discrete, vs. continuous, vs. hybrid models; (ii) deterministic, vs. non--deterministic, vs. stochastic models. AFSMs fall in the class of discrete non--deterministic systems. An exhaustive comparison of AFSMs with the aforementioned models is out of the scope of this section. We only mention that the discrete systems proposed in the literature well capture the network of the genes but lack in modeling the dynamics of each gene. 
In the following we show that AFSMs are appropriate to describe both the genes' network and the dynamics of each gene. Moreover, we show that the notion of  compositional bisimulation can lead to a sensible reduction in the size of the proposed model. \\
We consider the genetic regulatory system of the single--celled bacterium Escherichia Coli (\textit{E. coli}). In the sequel we only report basic facts about this regulatory system; the interested reader can refer to e.g. \cite{BiocBook,BiocBook1,Alon} for more details. 
\textit{E. coli} is a single--celled bacterium with a few million of proteins. During its life \textit{E. coli} encounters situations in which production of proteins is required. 
Each protein is produced by its gene. Each gene is a double strand of the Deoxyribonucleic acid (DNA) which encodes the information needed for the production of a specific protein. The transcription of a gene is the process by which  Ribonucleic acid (RNA) polymerase produces messenger RNA (mRNA) corresponding to the sequence of genetic code. The mRNA is then translated into a protein that is known as gene product. 
The genes whose activity is controlled in response to the needs of the cell, are called regulated genes. Regulated genes require special proteins called effectors or inductors which implement a kind of induction in the target gene. These proteins can bind to DNA and promote RNA transcription. 
When extracellular stimuli are perceived, such effectors promote RNA transcription and thus protein translation, as requested. 
Figure 4 describes a transcription network, representing about 10\% of the transcription interactions in \textit{E. coli}. Each node represents a gene. An edge (X,Y) indicates that the transcription factor encoded in X regulates operon Y. \\
\textit{E. coli} grows in moist soil containing salts which include a source of nitrogen, and a carbon source as glucose. 
The energy needed for biochemical reactions in \textit{E. coli} is derived from the metabolism of glucose and other secondary sugars including lactose, galactose and arabinose. For simplicity, in the following we focus on (only) the metabolic regulation of lactose, galactose and arabinose, see e.g. 
\cite{BiocBook1,BiocBook,Alon}. %,Megerle08,Morgan06}. 
When lactose is the sole source of carbon in the soil, three proteins are synthesized, which are necessary to metabolize lactose: 
\begin{itemize}
\item $\beta$--galactosidase (LacZ). This enzyme catalyzes splitting of lactose into glucose and galactose and catalyzes isomerization of lactose to allolactose.
\item	Lactose permease (LacY). It is located in the cytoplasmic membrane of \textit{E. coli} and is needed for the active transport of lactose into the cell.
\item Transacetylase (LacA). This enzyme modifies toxic molecules of lactose to facilitate their removal from the cell. 
\end{itemize}
When glucose is present, on average, only three molecules of $\beta$--galactosidase are present in the cell. This is because the genes of the three proteins are repressed by a protein encoded by gene LacI. After entering into the cell, lactose is converted into allolactose through a biochemical reaction that is catalyzed by one of the few copies of $\beta$--galactosidase. Then, allolactose binds to repressor LacI and after its dissociation, genes  LacZ, LacY and LacA are expressed, thus producing a 1000--fold increase in the concentration of $\beta$--galactosidase. 
As already mentioned, for lactose to be metabolized, two conditions are needed: presence of lactose and absence of glucose. The latter is perceived by the cell via the cyclic adenosine monophosphate (cAMP) that is a molecule whose concentration is inversely proportional to that of glucose. This molecule acts as a coactivator in respect of an activator protein, called cAMP receptor protein (CRP). When glucose is absent, the cAMP--CRP complex binds to a specific site near the promoter of the genes for LacZ, LacY and LacA and increases 50 times the transcription of their mRNA. 
Metabolic regulation of lactose can be formalized as an AFSM (see Figure 6 (Left Panel)) which involves proteins CRP, LacI, LacZ, LacY and LacA. 
We start by describing the FSM modeling the protein complex CRP--cAMP (Figure 5(a)). Complex CRP--cAMP switches from the inactive state $1$ to the active state $2$ when input $\{cAMP\}$, signaling absence of glucose, is perceived. As soon as glucose is perceived by the cell, the FSM switches to the initial state $1$. 
Similarly we can represent the evolution of protein LacI. The corresponding FSM is depicted in Figure 5(c). It consists of three states: state $7$, modeling that the protein is disabled, state $8$, modeling activation of the protein, and state $9$ modeling high activation of the protein. 
Transcribed proteins LacZ and LacY are illustrated in Figures 5(b)(d). We do not report the FSM of LacA because the mechanism by which LacA reacts to external stimuli is the same as the one of LacY; hence, we assume $LacY \cong LacA$. 
If regulator proteins CRP and LacI are disabled, proteins LacZ and LacY are in their basal states $3$ and $10$, respectively. Both LacZ and LacY switch from states $3$ and $10$ to low transcription states $4$, $5$ and respectively $11$, $12$ if only one of proteins CRP and LacI are activated or equivalently, if either CRP is in state $2$ or LacI is in state $8$. Finally, LacZ and LacY switch to high transcription states $6$ and $13$ if both CRP and LacI are in states $2$ and $8$, respectively. Moreover, when LacZ is in state $6$, it induces a transition in LacI from state $8$ (modeling activation of protein LacI) to state $9$ (modeling high activation of protein LacI). \\
The regulatory mechanism for the production of proteins capable of recruiting galactose and arabinose is similar to that of lactose. In particular, the galactose system involves proteins GalS, GalE, GalT, GalK. Figures 5(e)(f) reports FSM modeling of GalE and GalS respectively. FSMs of proteins GalT and GalK are not reported because the mechanism by which proteins GalT and GalK react to external stimuli is the same as the one of GalE, from which we assume $GalE \cong GalT \cong GalK$. 
The arabinose system involves proteins AraA, AraB, AraC, AraD, AraE, AraF, AraG and AraK. 
Figures 5(g)(h)(i) reports FSM modeling of AraB, AraE and AraC respectively. FSMs of proteins AraA, AraC, AraD, AraF and AraG are not reported because the external behavior of proteins AraE, AraF, AraG and AraH can be considered as equivalent, which implies $AraE \cong AraF \cong Ara G \cong AraH$; similarly, the external behavior of proteins AraB, AraA and AraD can be considered as equivalent, which implies $AraB \cong AraA \cong AraD$. 
The obtained AFSM $\mathbb{A}$ is reported in Figure 6 (left panel). \\ 
We conclude this section by computing the minimal bisimilar FSM of $\mathbf{M}(\mathbb{A})$ through the algorithm illustrated in Section 4.1: 
\begin{itemize} 
\item The relation $\mathbb{R}^{\ast}(\mathbb{A},\mathbb{A})$ has been computed and the induced equivalence classes are: $\{AraC\}$, $\{CPR\}$, $\{GalS\}$, $\{LacZ\}$, $\{LacI\}$, $\{LacY,Lac A\}$, $\{GalE,GalT,GalK\}$, $\{AraA,AraB,AraD\}$ and $\{AraE,$ $AraF,AraG,AraH\}$. 
\item The quotient $\mathbf{A}_{\min}(\mathbb{A})$ has been computed and is illustrated in Figure 6 (right panel). 
\item By expanding $\mathbf{A}_{\min}(\mathbb{A})$ the FSM $\mathbf{M}(\mathbf{A}_{\min}(\mathbb{A}))$ is obtained, which consists of $55,296$ states. 
\item The relation $R^{\ast}(\mathbb{M}(\mathbf{A}_{\min}(\mathbb{A})),\mathbb{M}(\mathbf{A}_{\min}(\mathbb{A})))$ is the identity relation.
\item The quotient $\mathbf{M}_{\min}(\mathbb{M}(\mathbf{A}_{\min}(\mathbb{A})))$ coincides with $\mathbf{M}(\mathbf{A}_{\min}(\mathbb{A}))$.
\end{itemize}
The above computations required to run bisimulation algorithms on the collection of FSMs $M_{i}$ composing $\mathbb{A}$, whose sets of states sum up to $35$ states, and the FSM $M_{\mathbb{A}}$ induced by $\mathbb{A}$, whose states are $17$. A naive approach to compute $\mathbf{M}_{\min}(\mathbf{M}(\mathbb{A}))$, would apply bisimulation algorithms directly to $\mathbf{M}(\mathbb{A})$, which is composed of $4,831,838,208$ states. 

\section{Conclusion}
In this paper we introduced the class of Arenas of Finite State Machines. We also proposed the notion of compositional bisimulation that 
provides a method to assess bisimulation equivalence between AFSMs, without the need of expanding them to the corresponding FSMs and hence, without incurring in the state explosion problem. Future research direction is two--fold. From the theoretical point of view, we will focus on generalizations of the results presented here to non--flat systems exhibiting more general compositional features, as both parallel and sequential composition. From the point of view of the systems' biology application that we proposed, we will investigate the use of AFSMs for the formal analysis of such systems.

\textbf{Acknowledgement:} The authors thank Alberto Sangiovanni Vincentelli, Davide Pezzuti, Pasquale Palumbo and Letizia Giampietro for fruitful discussions on the topics of this paper.

\bibliographystyle{alpha}
\bibliography{biblio1}

\end{document}